\documentclass[twocolumn]{article}

\usepackage{times}
\usepackage{graphicx} 

\usepackage{algorithm}
\usepackage{algorithmic}

\usepackage{hyperref}

\usepackage{amsfonts}
\usepackage{amsthm}

\newtheorem{theorem}{Theorem}
\newtheorem{lemma}{Lemma}

\newcommand{\tab}{\hspace{1cm}}

\newcommand{\E}{\mathbb{E}}
\usepackage{color}

\newcommand{\myFigureWidth}{8.3cm}

\title{An Algorithm for Online K-Means Clustering}
\author{
Edo Liberty\thanks{edo@yahoo-inc.com, Yahoo Labs, New York, NY}
\and
Ram Sriharsha\thanks{harshars@yahoo-inc.com, Yahoo Labs, Sunnyvale, CA}
\and
Maxim Sviridenko\thanks{sviri@yahoo-inc.com,Yahoo Labs, New York, NY}
}
\date{}

\begin{document} 
\maketitle



\begin{abstract}
This paper shows that one can be competitive with the $k$-means objective while operating online.
In this model, the algorithm receives vectors $v_1,\ldots,v_n$ one by one in an arbitrary order. 
For each vector $v_t$ the algorithm outputs a cluster identifier before receiving $v_{t+1}$.
Our online algorithm generates $\tilde{O}(k)$ clusters whose $k$-means cost is $\tilde{O}(W^*)$ where $W^*$ is the optimal $k$-means cost using $k$ clusters.\footnote{The notation $\tilde{O}(\cdot)$ suppresses poly-logarithmic factors.}
We also show that, experimentally, it is not much worse than $k$-means++ while operating in a strictly more constrained computational model.
\end{abstract}


\section{Introduction}

One of the most basic and well-studied optimization models in unsupervised Machine Learning is $k$-means clustering. 
In this problem we are given the set $V$ of $n$ points (or vectors) in Euclidian space. 
The goal is to partition $V$ into $k$ sets called clusters $S_1,\dots, S_k$ and choose one cluster center $c_i$ for each cluster $S_i$ to minimize 
\[
\sum_{i=1}^{k} \sum_{v \in S_i} ||v-c_i||_2^2.
\]
In the standard offline setting, the set of input points is known in advance and the data access model is unrestricted.
Even so, obtaining provably good solutions to this problem is difficult. See Section~\ref{related}.

In the streaming model the algorithm must consume the data in one pass and is allowed to keep only a small (typically constant or poly-logarithmic in $n$) amount of information. 
Nevertheless, it must output its final decisions when the stream has ended. For example, the location of the centers for $k$-means.
This severely restricted data access model requires new algorithmic ideas. See Section~\ref{related} for prior art. 
Notice that, in the streaming model, the assignment of individual points to clusters may become available only in hindsight.

In contrast, the online model of computation does not allow to postpone clustering decisions. 
In this setting, an a priori unknown number of points arrive one by one in an arbitrary order.
When a new point arrives the algorithm must either put it in one of the existing clusters or open a new cluster (consisting of a single point). 
Note that this problem is conceptually non trivial even if one could afford unbounded computational power at every iteration.
This is because the quality of current choices depend on the unknown (yet unseen) remainder of the stream.

In this paper, we consider the very restricted setting in the intersection of these two models.
We require the algorithm outputs a single cluster identifier for each point online while using space and time at most poly-logarithmic in the length of the stream.
This setting is harder than the streaming model. 
On the one hand, any space efficient online algorithm is trivially convertible to a streaming algorithm.
One could trivially keep sufficient statistics for each cluster such that the centers of mass could be computed at the end of the stream.
The computational and space overhead are independent of the length of the stream.
On the other hand, the online problem resists approximation even in one dimension and $k=2$.

Consider the stream where $v_1 = 0$ and $v_2 = 1$ (acting as one dimensional vectors). 
Any online clustering algorithm must assign them to different clusters. 
Otherwise, the algorithm cost is $1/2$ and the optimal is cost is trivially $0$.
If the the algorithm assigns $v_1$ and $v_2$ to different clusters, the third point might be  $v_3 = c$ for some $c \gg 1$.
At this point, the algorithm is forced to assign $v_3$ to one of the existing clusters incurring cost of $\Omega(c)$ which is arbitrarily larger than the optimal solution of cost $1/2$.
This example also proves that any online algorithm with a bounded approximation factor (such as ours) must create strictly more than $k$ clusters. 

In this work we provide algorithms for both online $k$-means and {\it semi-online} $k$-means. 
In the {\it semi-online} model we assume having a lower bound, $w^*$, for the total optimal cost of $k$-means, $W^*$, as well as an estimate for $n$, the length of the stream. 
Algorithm~\ref{algSemiOnline} creates at most $$O(k\log n \log (W^*/w^*))$$ clusters in expectation and has an expected objective value of $O(W^*)$.
From a practical viewpoint, it is reasonable to assume having rough estimates for  $w^*$ and $n$.
Since the dependence on both estimates is logarithmic, the performance of the {\it semi-online} algorithm will degrade significantly only if our estimates are wrong by many orders of magnitude.
In the fully online model we do not assume any prior knowledge.
Algorithm~\ref{algFullyOnline} operates in that setting and opens a comparable number of clusters to Algorithm~\ref{algSemiOnline}. 
But, its approximation factor guarantee degrades by a $\log n$-factor. 

\subsection{Motivation}

In the context of machine learning, the results of $k$-means were shown to provide powerful unsupervised features \cite{CoatesNL11} on par, sometimes, with neural nets for example.
This is often referred to as (unsupervised) feature learning.
Intuitively, if the clustering captures most of the variability in the data, assigning a single label to an entire cluster should be pretty accurate.
It is not surprising therefore that cluster labels are powerful features for classification. 
In the case of online machine learning, these cluster labels must also be assigned online.
The importance of such an online $k$-means model was already recognized in machine learning community \cite{ChoromanskaM12, DasguptaCSE291}.  


For information retrieval, \cite{CharikarCFM97} investigated the incremental $k$-centers problem.
They argue that clustering algorithms, in practice, are often required to be online.
We observe the same at Yahoo. 
For example, when suggesting news stories to users, we want to avoid suggesting those that are close variants of those they already read.
Or, conversely, we want to suggest stories which are a part of a story-line the user is following.
In either scenario, when Yahoo receives a news item, it must immediately decide what cluster it belongs to and act accordingly.

\subsection{Prior Art}\label{related}

In the offline setting where the set of all points is known in advance, Lloyd's algorithm \cite{Lloyd82} provides popular heuristics. 
It is so popular that practitioners often simply refer to it as $k$-means. 
Yet, only recently some theoretical guaranties were proven for its performance on ``well clusterable" inputs \cite{OstrovskyRSS12}.
The $k$-means++ \cite{ArthurV07} algorithm provides an expected  $O(\log(k))$ approximation or an efficient seeding algorithm. 
A well known theoretical algorithm is due to Kanungo et al. \cite{KanungoMNPSW02}. 
It gives a constant approximation ratio and is based on local search ideas popular in the related area of design and analysis of algorithms for facility location problems, e.g., \cite{AryaGKMMP04}.
Recently, \cite{AggarwalDK09} improved the analysis of \cite{ArthurV07} and gave an adaptive sampling based algorithm with constant factor approximation to the optimal cost.
In an effort to make adaptive sampling techniques more scalable, \cite{BahmaniMVKV12} introduced $k$-means$\|$ which reduces the number of passes needed over the data and enables improved parallelization.

The streaming model was considered by \cite{AilonJM09} and \cite{SWM} and later by \cite{AckermannMRSLS12}.
They build upon adaptive sampling ideas from \cite{ArthurV07,BahmaniMVKV12} and branch-and-bound techniques from \cite{GuhaMMMO03}.

The first (to our knowledge) result in online clustering dates back the $k$-centers result of \cite{CharikarCFM97}.
For $k$-means an Expectation Maximization (EM) approach was investigated by \cite{LiangK09}. 
Their focus was on online EM as a whole but their techniques include online clustering.
They offer very encouraging results, especially in the context of machine learning. To the best of our understanding, however, their techniques do not extend to arbitrary input sequences.
In contrast, the result of \cite{ChoromanskaM12} provides provable results for the online setting in the presence of base-$k$-means algorithm as experts.


A closely related family of optimization problems is known as facility location problems. 
Two standard variants are the uncapacitated facility location problem (or the simple plant location problem in the Operations Research jargon) and the k-median problem. 
These problems are well-studied both from computational and theoretic viewpoints (a book \cite{DreznerH02} and a survey \cite{Vygen05} provide the  background on some of the aspects in this area). 
Meyerson \cite{Meyerson01} suggested a simple and elegant algorithm for the online uncapacitated facility location with competitive ratio of $O(\log n)$. 
Fotakis \cite{Fotakis08} suggested a primal-dual algorithm with better performance guarantee of $O(\log n/\log \log n)$. 
Anagnostopoulos et al. \cite{AnagnostopoulosBUH04} considered a different set of algorithms based on hierarchical partitioning of the space and obtained similar competitive ratios. 
The survey \cite{Fotakis11} summarizes the results in this area. As a remark, \cite{CharikarCFM97}  already considered connections between facility location problems and clustering.
Interestingly, their algorithm is often referred to as ``the doubling algorithm" since the cluster diameters double as the algorithm receives more points. 
In our work the facility location cost is doubled which is technically different but intuitively related.

\section{Semi-Online $k$-means Algorithm}\label{secSemiOnline}
We begin with presenting the {\it semi-online} algorithm. 
It assumes knowing the number of vectors $n$ and some lower bound $w^*$ for the value of the optimal solution.
These assumptions make the algorithm slightly simpler and the result slightly tighter.
Nevertheless, the {\it semi-online} online already faces most of the challenges faced by the {\it fully online} version.
In fact, proving the correctness of the online algorithm (Section \ref{secFullyOnline}) would require only minor adjustments to the proofs in this section.

The algorithm uses ideas from the online facility location algorithm of Meyerson \cite{Meyerson01}.
The intuition is as follows; think about $k$-means and a facility location problem where the service costs are squared Euclidean distances.
For the facility cost, start with $f_1$ which is known to be too low. 
By doing that the algorithm is ``encouraged" to open many facilities (centers) which keeps the service costs low.
If the algorithm detect that too many facilities were opened, it can conclude that the current facility cost is too low. 
It therefore doubles the facility cost of opening future facilities (centers).
It is easy to see that the facility cost cannot be doubled many times without making opening new clusters prohibitively expensive.
In Algorithm \ref{algFullyOnline} we denote the distance of a point $v$ to a set $C$ as $D(v, C) = \min_{c \in C}\|v - c\|$.
As a convention, if $C = \emptyset$ then $D(v, C) = \infty$ for any $v$.

\begin{algorithm}
\begin{algorithmic}
\STATE {\bf input:} $V$, $k$, $w^*$, $n$
\STATE $C \gets \emptyset$
\STATE $r \gets 1$; $q_1 \gets 0$; $f_1 \gets w^*/k\log(n)$
\FOR {$v \in V$}
	\STATE {\bf with probability} $p = \min(D^2(v, C)/f_r,1)$
	\STATE \tab $C \gets C \cup \{v\}; q_r\gets q_r+1$
	\IF {$q_r \ge 3 k (1+ \log_{} (n))$}
		\STATE $r \gets r+1$; $q_r \gets 0$; $f_r \gets 2\cdot f_{r-1}$
	\ENDIF
	\STATE {\bf yield:} $c = \arg\min_{c \in C}\|v - c\|^2$
\ENDFOR
\caption{semi-online $k$-means algorithm}\label{algSemiOnline}
\end{algorithmic}
\end{algorithm}

Consider some optimal solution consisting of clusters $S_1^*,\dots,S^*_k$ with cluster centers $c_1^*,\dots,c^*_k$. Let
$$W^*_i=\sum_{v\in S^*_i} ||v-c_i^*||_2^2$$
be the cost of the $i$-th cluster in the optimal solution and 
 $W^*=\sum_{i=1}^kW^*_i$ be the value of the optimal solution. 
Let $A^*_i$ be the average squared distance to the cluster center from a vectors in the $i$-th optimal cluster. 
\[
A^*_i = \frac{1}{|S^*_i|}\sum_{v \in C^*_i} \|v - c_i^*\|_2^2 = \frac{W^*_i}{|S^*_i|}\ .
\]
We define a partition of the cluster $S_i^*$ into subsets that we call {\it rings}: 
\[
S^{*}_{i,0}=\{v\in S^*_i: ||v-c^*_i||^2_2\le A^*_i\}
\]
and for $1 \le \tau \le \log n$
\[
S^{*}_{i,\tau}=\left\{v\in S^*_i: ||v-c^*_i||^2_2\in ( 2^{\tau-1}A^*_i,2^{\tau}A^*_i]\right\} \ .
\]

Note that we consider only values of $\tau \le \log n$ since $S^{*}_{i,\tau}=\emptyset $  for $\tau > \log_{} (|S_i^*|)$.
To verify  assume the contrary and compute $A^*_i$.

\begin{theorem}\label{thmNumberOfClusterSemiOnline}
Let $C$ be the set of clusters defined by Algorithm~\ref{algSemiOnline}. 
Then 
$$\E[ |C| ]=O\left(k \log_{} n \log_{} \frac{W^*}{w^*}\right) \ .$$
\end{theorem}
\begin{proof}
Consider the phase $r'$ of the algorithm where, for the first time
$$f_{r'} \ge \frac{W^*}{k\log_{} n} \ .$$
The initial facility cost $f_1$ is doubled at every phase during each of which the algorithm creates $3k (1+\log_{} n)$ clusters.
The total number of clusters opened before phase $r'$ is trivially upper bounded by $3k (1+\log_{} n)\log_{}\frac{ f_{r'}}{f_1}$. 
Which is, in turn, $O(k \log_{} n \log_{}\frac{W^*}{w^*})$ by the choice of $f_1$.

Bounding the number of centers opened during and after phase $r'$ is more complicated.
%
Denote by $S^{*}_{i,\tau,r}$ the set of points in the ring $S^{*}_{i,\tau}$ that our algorithm encounters during phase $r$.
The expected number of clusters initiated by vectors in the ring $S^{*}_{i,\tau}$ during phases $r',\dots,R$ is at most
$$1+\sum_{r\ge r'}\frac{4\cdot 2^{\tau}A^*_i}{f_r}|S^{*}_{i,\tau,r}| \ . $$
This is because once we open the first cluster with a center at some $v\in S^{*}_{i,\tau}$ the probability of 
opening a cluster for each subsequent vector  $v'\in S^{*}_{i,\tau}$ is upper bounded by 
$$\frac{||v-v'||^2_2}{f_r}\le \frac{2 ||v-c^*_i||^2_2+2||v'-c^*_i||^2_2}{f_r}\ \le \frac{4\cdot 2^{\tau}A^*_i}{f_r}$$ 
by the (squared) triangle inequality for $v,v' \in S^*_{i,\tau}$.

Therefore the expected number of clusters chosen from $S^*_i$ during and after phase $r'$ is at most
\begin{eqnarray*}
&& \sum_{\tau\ge0}\left(1+\sum_{r\ge r'}\frac{4\cdot 2^{\tau} A^*_i}{f_r}|S^{*}_{i,\tau,r}|\right) \\
&& \le 1+ \log n + \sum_{\tau \ge 0} 4\cdot2^{\tau}A_i^* \sum_{r\ge r'} \frac{|S^{*}_{i,\tau,r}|}{f_r}\\
&& \le 1+ \log n + \frac{4}{f_{r'}}\sum_{\tau \ge 0} 2^{\tau} A_i^*|S^{*}_{i,\tau}| \\
&& \le 1+ \log n + \frac{4}{f_{r'}}A^*_i|S^{*}_{i,0}|+\frac{8}{f_{r'}}\sum_{\tau \ge 1} 2^{\tau-1}A^*_i|S^{*}_{i,\tau}|\\
&& \le 1+ \log n + \frac{4}{f_{r'}}A^*_i|S^{*}_{i}|+\frac{8}{f_{r'}}\sum_{\tau \ge 1} \sum_{v \in S^{*}_{i,\tau}} \|v - c^*_i\|^2\\
&& \le 1+ \log n + \frac{4}{f_{r'}}W_i^*+\frac{8}{f_{r'}}W_i^*\\
&& \le 1+ \log n + \frac{12W^*_i}{f_{r'}}.
\end{eqnarray*}
Summing up over all $i=1,\dots, k$ using $\sum_i W_i^* = W^*$ we obtain that the expected number of centers chosen during phases $r',\ldots,R$
is at most 
\begin{equation}\label{centersAfterPhase}
k(1+\log n) + 12W^*/f_{r'} \ .
\end{equation}
Substituting $f_{r'} \ge W^*/k \log n$ completes the proof of the theorem.
\end{proof}

Before we estimate the expected cost of clusters opened by our online algorithm we prove the following technical lemma.
\begin{lemma}\label{random}
We are given a sequence $X_1,\dots, X_n$ of $n$ independent experiments. 
Each experiment succeeds with probability $p_i\ge \min\{A_i/B,1\}$ where $B\ge 0$ and $A_i\ge 0$ for $i=1,\dots,n$.
Let $t$ be the (random) number of sequential unsuccessful experiments, then:
$$\E\left[\sum_{i=1}^{t}A_i\right]\le B.$$
\end{lemma}
\begin{proof}
Let $n'$ be the maximal index for which $p_i < 1$ for all $i\le n'$.
\begin{eqnarray*}
\E\left[\sum_{i=1}^{t}A_i\right] &=& \sum_{i=1}^{n'}A_i \Pr[t \ge i] \\
&\le& \sum_{i=1}^{n'}A_i\prod_{j=1}^i\left(1-\frac{A_j}{B}\right)\\
&\le& B \sum_{i=1}^{n'}\frac{A_i}{B}\prod_{j=1}^{i-1}\left(1-\frac{A_j}{B}\right)\\
&\le& B.
\end{eqnarray*}
The last inequality uses $A_i/B \le p_i < 1$ for $i \le n'$.
\end{proof}

\begin{theorem}\label{thmApproxSemiOnline}
Let $W$ be the cost of the online assignments of Algorithm \ref{algSemiOnline} and $W^*$ the optimal $k$-means clustering cost. Then 
$$\E[W]=O(W^*) \ .$$
\end{theorem}
\begin{proof}
Consider the service cost of vectors in each ring $S^{*}_{i,\tau}$ in two separate stages. Before a vector from the ring is chosen to start a new cluster and after.
Before a center from  $S^{*}_{i,\tau}$ is chosen each vector $v \in S^{*}_{i,\tau}$ is chosen with probability $p \ge \min\{d^2(v,C)/f_R,1\}$. 
Here, $C$ is the set of centers already chosen by the algorithm before encountering $v$.
If $v$ is not chosen the algorithm incurs a cost of $d^2(v,C)$. 
By Lemma \ref{random} the expected sum of these costs is bounded by $f_R$.
Summing over all the rings we get a contribution of $O(f_R k\log n )$.

After a vector $v \in S^{*}_{i,\tau}$ is chosen to start a new cluster, the service cost of each additional vector $v'$ is at most $\|v - v'\|^2 \le 4\cdot 2^{\tau}A_i^*$.
Summing up over all vectors and rings, this stage contributes are most $4\sum_{i}\sum_{\tau} \cdot 2^{\tau}A_i^*\cdot |S^{*}_{i,\tau}| \le 12 W^*$ to the cost of our solution.
All in all, the expected online $k$-means cost  is bounded by $$\E[W] = O(f_R k\log n  + W^*) \ . $$

We now turn to estimating $\E[f_R]$. Consider the first  phase $r''$ of the algorithm such that 
\[
f_{r''}\ge \frac{36W^*}{k(1+\log_{} n)} \ .
\]
By Equation~\ref{centersAfterPhase} the expected number of clusters opened during and after phase $r''$ is at most $k(1+\log_{}n) +12W^*/f_{r''} \le \frac43 k(1+\log_{}n)$. 
By Markov's inequality the probability of opening more than $3 k(1+\log_{}n) $ clusters is at most $4/9$. 
Therefore, with probability at least $5/9$ the algorithm will conclude while at phase $r''$.

%
Let $p$ be the probability that our algorithm terminates before round $r''$. 
Since the probability of concluding the execution at each of the rounds after $r''$ is at least $5/9$ we derive an upper bound
\begin{eqnarray*}
\E[f_R]&\le& p f_{r''-1} + (1-p)\sum_{r=r''}^{+\infty}f_r\cdot \frac59 \cdot \left(\frac49\right)^{r-r''}\\
&<&f_{r''} + f_{r''}\cdot \frac59 \sum_{i=0}^{+\infty}2^i\cdot \left(\frac49\right)^i = O(f_{r''})
\end{eqnarray*}
Combining $\E[f_R] = O(f_{r''})$ with  our choice of $f_{r''} = O(\frac{W^*}{k(1+\log_{} n)})$ and our previous observation that $\E[W] = O(f_R k\log n  + W^*)$ completes the proof.
\end{proof}

\section{Fully Online $k$-means Algorithm}\label{secFullyOnline}

Algorithm \ref{algFullyOnline} is fully online yet it defers from Algorithm \ref{algSemiOnline} in only a few aspects.
First, since $n$ is unknown, the initial facility cost and the doubling condition cannot depend on it.
Second, it must generate its own lower bound $w^*$ based on a short prefix of points in the stream.
Note that $w^*$ is trivially smaller that $W^*$.
Any clustering of $k+1$ points must put at least two points in one cluster, incurring a cost of $\|v -v'\|^2/2 \ge \min_{v,v'}\|v -v'\|^2/2$.

\begin{algorithm}
\begin{algorithmic}
\STATE {\bf input:} $V$, $k$
\STATE $C \gets$ the first $k+1$ distinct  vectors in $V$; and $n=k+1$
\STATE (For each of these {\bf yield} itself as its center)
\STATE $w^{*} \gets \min_{v,v' \in C} \|v -v'\|^2/2$
\STATE $r \gets 1$; $q_1 \gets 0$; $f_1 = w^*/k$
\FOR {$v \in$ the remainder of  $V$}
	\STATE $n \gets n+1$
	\STATE {\bf with probability} $p = \min(D^2(v, C)/f_r,1)$
	\STATE \tab $C \gets C \cup \{v\}; q_r\gets q_r+1$
	\IF {$q_r \ge 3 k (1+ \log_{} (n))$}
		\STATE $r \gets r+1$; $q_r \gets 0$; $f_r \gets 2\cdot f_{r-1}$
	\ENDIF
	\STATE {\bf yield:} $c = \arg\min_{c \in C}\|v - c\|^2$
\ENDFOR
\caption{Online $k$-means algorithm}\label{algFullyOnline}
\end{algorithmic}
\end{algorithm}

\begin{theorem}\label{thmNumberOfClusterFullyOnline}
Let $C$ be the set of clusters defined by Algorithm~\ref{algFullyOnline}.  
Then 
$$\E[ |C| ]=O\left(k \log_{} n \log_{} \frac{W^*}{w^*}\right) = O\left(k \log_{} n \log_{} \gamma n \right) \ .$$
\end{theorem}
Here  $\gamma = \frac{\max_{v,v'}\|v-v'\|}{\min_{v,v'}\|v -v'\|}$ is the dataset ``aspect ratio''.
\begin{proof}
Intuitively, for the same lower bound $w^*$ Algorithm~\ref{algFullyOnline} should create fewer centers than Algorithm~\ref{algSemiOnline} since its initial facility cost is higher and it is doubled more frequently.
This intuition is made concrete by retracing the proof of Theorem~\ref{thmNumberOfClusterFullyOnline} to show $$\E[ |C| ] = O\left(k \log_{} n \log_{} \frac{W^*}{w^*}\right) \ .$$

To get a handle on the value of $W^*/w^*$, observe that $W^* \le n \max_{v,v'}\|v-v'\|^2$. 
Combining this with the definition of $\gamma$ we get $\log (W^*/w^*) = O(\log \gamma n)$.
\end{proof}

\begin{theorem}
Let $W$ be the cost of the online assignments of Algorithm \ref{algFullyOnline} and $W^*$ the optimal $k$-means clustering cost. Then 
$$\E[W]=O(W^* \log n ) \ .$$
\end{theorem}
\begin{proof}  
We start by following the argument of the proof of Theorem~\ref{thmApproxSemiOnline} verbatim.
We arrive at the conclusion that $$\E[W] = O(f_R k\log n  + W^*) $$
where $f_R$ is the final facility cost of the algorithm and $R$ is its last phase.
Showing that $\E[f_R] = O(W^*/k)$ will therefore complete the proof.

Consider any phase $r \ge r''$ of the algorithm where $r''$ is the smallest index such that 
$$f_{r''} \ge \frac{36W^*}{k}.$$
Let $n_{r}$ be the number of points from the input the algorithm went through by the end of phase $r$.
Let $q_r$ be the number of clusters opened during phase $r$ and $q'_r$ the number those who are {\it not} the first in their ring.
$$q_r \le k\log(1+\log n_r) + q'_r$$ 
The term $k\log(1+\log n_r)$ is an upper bound on the number of rings at the end of stage $r$.
We pessimistically count at most one (the first) cluster from each such ring. 
Following the argument in the proof of Theorem \ref{thmNumberOfClusterSemiOnline} that lead us to Equation~(\ref{centersAfterPhase}) we conclude $\E[q'_r] \le 12W^*/f_r$.

Algorithm~\ref{algFullyOnline} only advances to the next phase if $q_r \ge 3 \log(1+\log n_r)$ which requires $q'_r \ge 2k(1+\log n_r)$. 
By Markov's inequality and the fact that $\E[q'_r] \le 12W^*/f_r \le k/3$ the probability of reaching the next phase is at most $1/6$.

We now estimate $\E[f_R]$. Let $p$ be the probability that our algorithm finishes before round $r''$. We have
\begin{eqnarray*}
\E[f_R]&\le& p f_{r''-1} + (1-p)\sum_{r=r''}^{+\infty}f_r\cdot \frac56 \cdot \left(\frac16\right)^{r-r''}\\
&\le&f_{r''} + f_{r''}\cdot \frac56\sum_{i=0}^{+\infty}2^i\cdot \left(\frac16\right)^i = O(f_{r''})\\
\end{eqnarray*}
Since $f_{r''} = O(W^*/k)$ the proof is complete.
\end{proof}

\section{Experimental Analysis of the Algorithm}
\subsection{Practical modifications to the algorithm}
While experimenting with the algorithm, we discovered that some $\log$ factors were, in fact, too pessimistic in practice.
We also had to make some pragmatic decisions about, for example, how to set the initial facility cost.
As another practical adjustment we introduce the notion of $k_{target}$ and $k_{actual}$.
The value of $k_{target}$ is the number of clusters we would like the algorithm to output while $k_{actual}$ is the actual number of clusters generated.
Internally, the algorithm operates with a value of $k = \lceil(k_{target}-15)/5\rceil$. 
This is a heuristic (entirely ad-hoc) conversion that compensates for the $k_{actual}$ being larger than $k$ by design.

\begin{algorithm}
\begin{algorithmic}
\STATE {\bf input:} $V$, $k_{target}$
\STATE $k = \lceil(k_{target}-15)/5\rceil$
\STATE $C \gets$ the first $k+10$ vectors in $V$
\STATE (For each of these {\bf yield} itself as its center)
\STATE $w^*\gets$ half the sum of the $10$ smallest squared distances of points in $C$ to their closest neighbor.
\STATE $r \gets 1$; $q_1 \gets 0$; $f_1 \gets w^*$
\FOR {$v \in$ the remainder of  $V$}
	\STATE {\bf with probability} $p = \min(D^2(v, C)/f_r,1)$
	\STATE \tab $C \gets C \cup \{v\}; q_r\gets q_r+1$
	\IF {$q_r \ge k $}
		\STATE $r \gets r+1$; $q_r \gets 0$; $f_r \gets 10\cdot f_{r-1}$
	\ENDIF
	\STATE {\bf yield:} $c = \arg\min_{c \in C}\|v - c\|^2$
\ENDFOR
\STATE $k_{actual} \gets |C|$
\caption{Online $k$-means algorithm}\label{algExperimental}
\end{algorithmic}
\end{algorithm}

\subsection{Datasets}
To evaluate our algorithm we executed it on $12$ different datasets.
All the datasets that we used are conveniently aggregated on the LibSvm website \cite{libsvmData} and on the UCI dataset collection \cite{UCIdata}.
Some basic information about each dataset is given in Table~\ref{table1}. 

\newcommand{\ttt}[2]{ \parbox{#1pt}{\vspace{0.1cm} \noindent  #2 \vspace{0.1cm}}}
\begin{table}[h!]
\begin{center}
\begin{tabular}{|c|c|c|c|} \hline
Dataset		&	$nnz$		& $n$	&	$d$	\\ \hline
20news-binary	&	2.44E+6	&	1.88E+4	&	6.12E+4	\\ \hline
adult			&	5.86E+5	&	4.88E+4	&	1.04E+2	\\ \hline
ijcnn1		&	3.22E+5	&	2.50E+4	&	2.10E+1	\\ \hline
letter			&	2.94E+5	&	2.00E+4	&	1.50E+1	\\ \hline
magic04		&	1.71E+5	&	1.90E+4	&	9.00E+0 	\\ \hline
maptaskcoref	&	6.41E+6	&	1.59E+5	&	5.94E+3	\\ \hline
nomao		&	2.84E+6	&	3.45E+4	&	1.73E+2	\\ \hline
poker		&	8.52E+6	&	9.47E+5	&	9.00E+0	\\ \hline
shuttle		&	2.90E+5	&	4.35E+4	&	8.00E+0	\\ \hline
skin			&	4.84E+5	&	2.45E+5	&	2.00E+0	\\ \hline
vehv2binary	&	1.45E+7	&	2.99E+5	&	1.04E+2	\\ \hline
w8all			&	7.54E+5	&	5.92E+4	&	2.99E+2	\\ \hline
\end{tabular}
\end{center}
\caption{The table gives some basic information about the datasets we experimented with. 
The column under $nnz$ gives the number of non zero entries in the entire dataset, $n$ the number of vectors and $d$ their dimension.
Much more information is provided on LibSvm website \cite{libsvmData} and in the UCI dataset collection \cite{UCIdata}.
}
\label{table1}
\end{table}%

Feature engineering for the sake of online learning is one of the motivations for this work.
For that reason, we apply standard stochastic gradient descent linear learning with the squared loss on these data.
Once with the raw features and once with the $k$-means features added. 
In some cases we see a small decrease in accuracy due to slower convergence of the learning on a larger feature set.
This effect should theoretically be nullified in the presence of more data.
In other cases, however, we see a significant uptick in classification accuracy. This is in agreement with prior observations \cite{CoatesNL11}.

\begin{table}[h!]
\begin{center}
\begin{tabular}{|c|c|c|} \hline
Dataset	& \ttt{60}{Classification \\ accuracy  with  \\ raw features}  	& \ttt{70}{Classification \\ accuracy with \\ $k$-means features}	\\ \hline
20news		& 0.9532										& 0.9510	\\ \hline 
adult			& 0.8527										& 0.8721	\\ \hline
ijcnn1		& 0.9167										& 0.9405	\\ \hline 
letter			& 0.7581										& 0.7485	\\ \hline
magic04		& 1.0000										& 1.0000	\\ \hline
maptaskcoref	& 0.8894  										& 0.8955	\\ \hline
nomao		& 0.5846										& 0.5893	\\ \hline
poker		& 0.5436										& 0.6209	\\ \hline
shuttle		& 0.9247										& 0.9973	\\ \hline
skin			& 0.9247										& 0.9988	\\ \hline
vehv2binary	& 0.9666										& 0.9645	\\ \hline
w8all			& 0.9638										& 0.9635	\\ \hline
\end{tabular}
\end{center}
\caption{Corroborating the observations of \cite{CoatesNL11} we report that 
adding $k$-means feature, particularly to low dimensional datasets, is very beneficial for improving classification
accuracy. Indeed enabling online machine leaning to take advantage of this phenomenon is one of the motivations for performing $k$-means online.}
\label{table2}
\end{table}%

\subsection{The number of online clusters}
One of the artifacts of applying our online $k$-means algorithm is that the number of clusters is not exactly known a priory.
But as we see in Figure~\ref{fig1}, the number of resulting clusters is rather predictable and controllable.
Figure~\ref{fig1} gives the ratio between the number of clusters output by the algorithm, $k_{actual}$, and the specified target $k_{target}$.
The results reported are mean values of $3$ runs for every parameter setting. 
The observed standard deviation of $k_{actual}$ is typically in the range $[0,3]$ and never exceeded $0.1\cdot k_{target}$ in any experiment.
Figure~\ref{fig1} clearly shows that the ratio $k_{actual}/k_{target}$ is roughly constant and close $1.0$. 
Interestingly, the main differentiator is the choice of dataset.

\begin{figure}[h!]
\begin{center}
\includegraphics[width=\myFigureWidth]{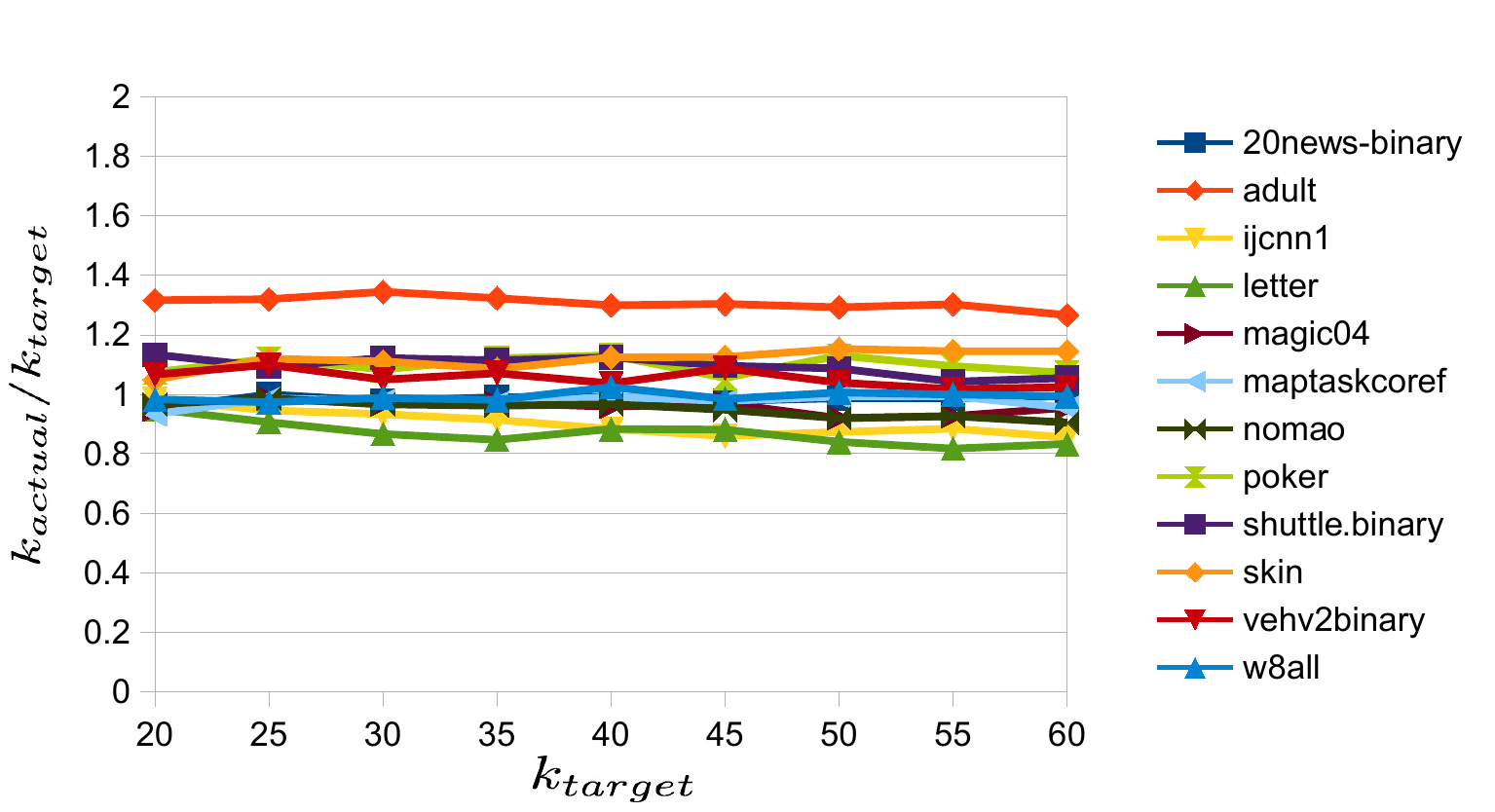}
\caption{The figure gives the ratio $k_{actual}/k_{target}$ on the $y$-axis as a function of $k_{target}$ on the $x$-axis.
The value $k_{target}$ is given to the algorithm as input and $k_{actual}$ is the resulting cardinality of the center set $C$.
We clearly see that this ratio is roughly constant and close $1$. Interestingly, the main differentiator is the dataset itself and not the value of $k_{target}$.
}
\label{fig1}
\end{center}
\end{figure}

\subsection{Online clustering cost}
Throughout this section, we measure the online $k$-means clustering cost with respect to different baselines. We report averages of at least 3 different independent executions for every parameter setting.
In Figure~\ref{fig2} the reader can see the online $k$-means clustering cost for the set of centers chosen online by our algorithm for different values of $k_{target}$ and different datasets.
For normalization, each cost is divided by $f_0$, the sum of squares of all vector norms in the dataset (akin to the theoretical $k$-means cost of having one center at the origin).
Note that some datasets are inherently unclusterable. Even using many cluster centers, the $k$-means objective does not decrease substantially.
Nevertheless, as expected, the $k$-means cost obtained by the online algorithm, $f_{online}$, decreases as a function of $k_{target}$. 

\begin{figure}[h!]
\begin{center}
\includegraphics[width=\myFigureWidth]{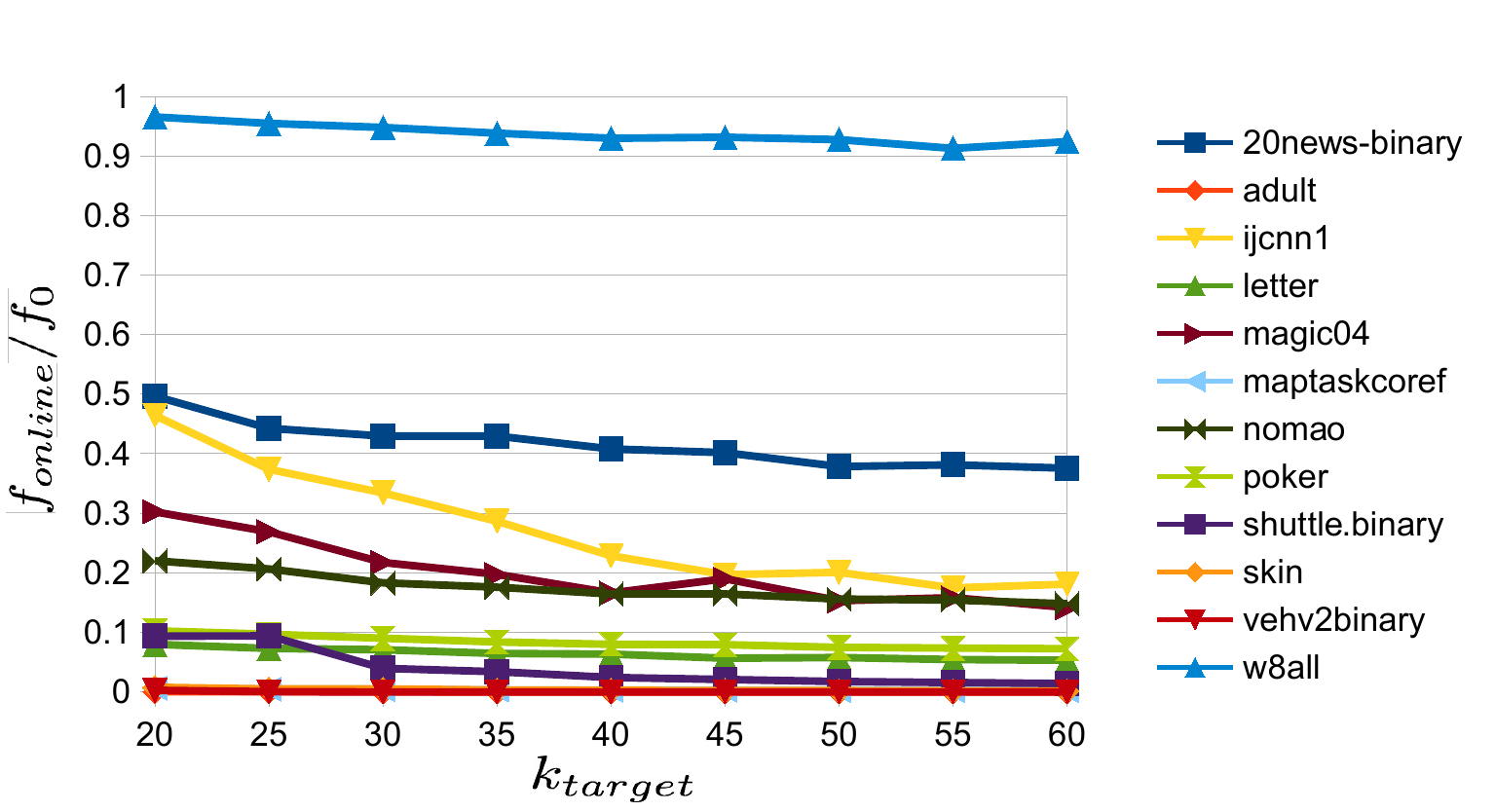}
\caption{Online $k$-means clustering cost ($f_{online}$) as a function of $k_{target}$ for the different datasets.
For normalization, each cost is divided by $f_0$, the sum of squares of all vector norms in the dataset (akin to the $k$-cost of once center in the origin).}
\label{fig2}
\end{center}
\end{figure}

The monotonicity of $f_{online}$ with respect to $k_{target}$ is unsurprising.
In Figure~\ref{fig3} we plot the ratio $f_{online}/f_{random}$ as a function of $k_{target}$.
Here, $f_{random}$ is the sum of squared distances of input points to $k_{actual}$  input points chosen uniformly at random (as centers).
Note that in each experiment the number of clusters used by the random solution and online $k$-means is identical, namely, $k_{actual}$.
Figure~\ref{fig3} illustrates something surprising. The ratio between the costs remains relatively fixed per dataset and almost independent to $k_{target}$.
Put differently, even when the $k$-means cost is significantly lower than picking $k$ random centers, they improve in similar rates as $k$ grows.

\begin{figure}[h!]
\begin{center}
\includegraphics[width=\myFigureWidth]{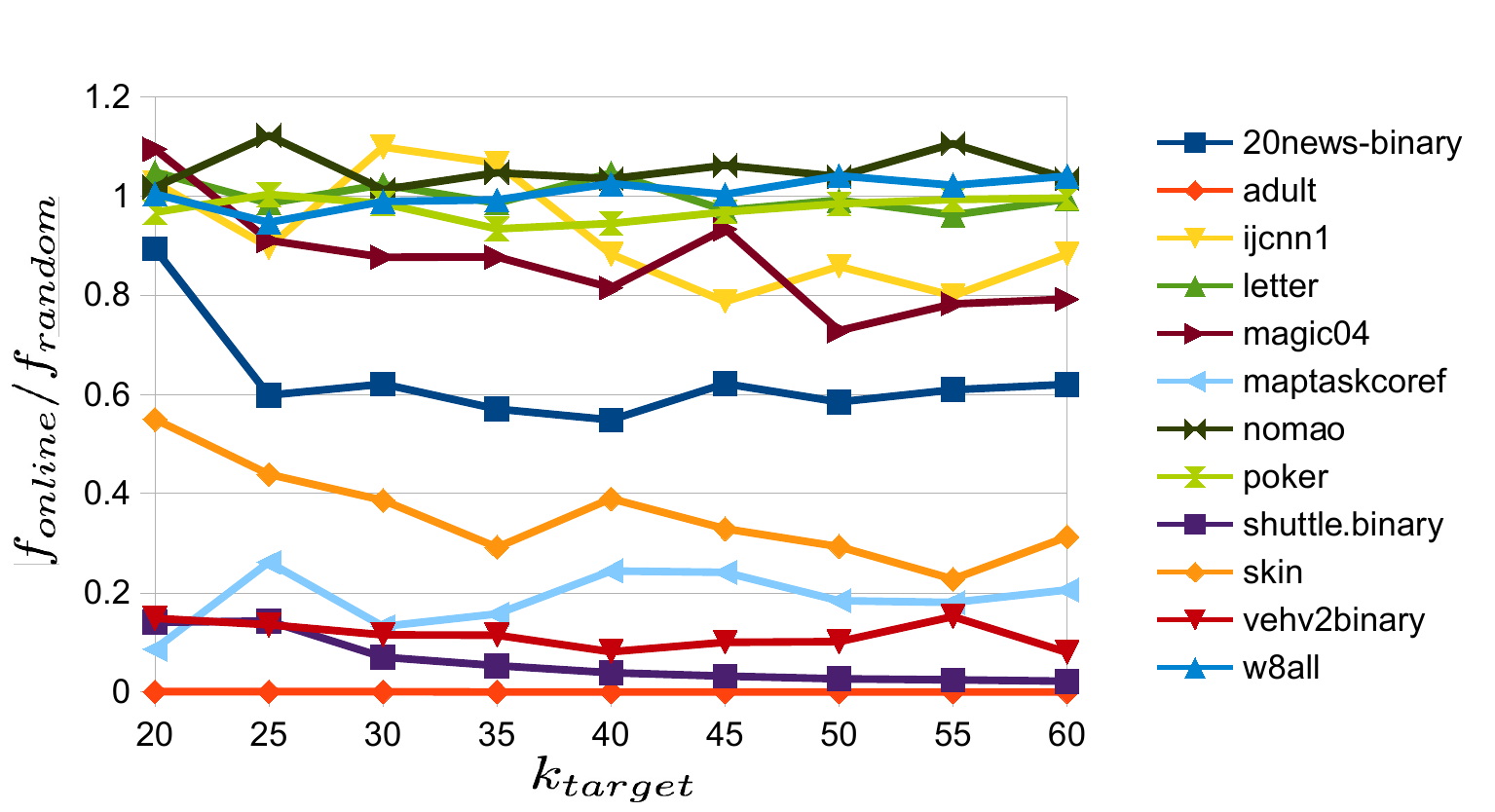}
\caption{On the $y$-axis, the value of $f_{online}$ divided by $f_{random}$. 
The latter is the cost of choosing, uniformly at random, as many cluster centers (from the data) as the online algorithm did.
A surprising observation is that this ratio is almost constant for each dataset and almost independent of $k_{target}$ (on the $x$-axis).}
\label{fig3}
\end{center}
\end{figure}

The next experiment compares online $k$-means to $k$-means++.
For every value of $k_{target}$ we ran online $k$-means to obtain both $f_{online}$ and $k_{actual}$.
Then, we invoke $k$-means++ using $k_{actual}$ clusters and computed its cost, $f_{kmpp}$.
This experiment was repeated $3$ times for each dataset and each value of $k_{target}$.
The mean results are reported in Figure~\ref{fig4}. 
Unsurprisingly, $k$-means++ is usually better in terms of cost. But, the reader should keep in mind that
$k$-means++ is an \emph{offline} algorithm that requires $k$ passes over the data compared with the online computational model of our algorithm.

\begin{figure}[h!]
\begin{center}
\includegraphics[width=\myFigureWidth]{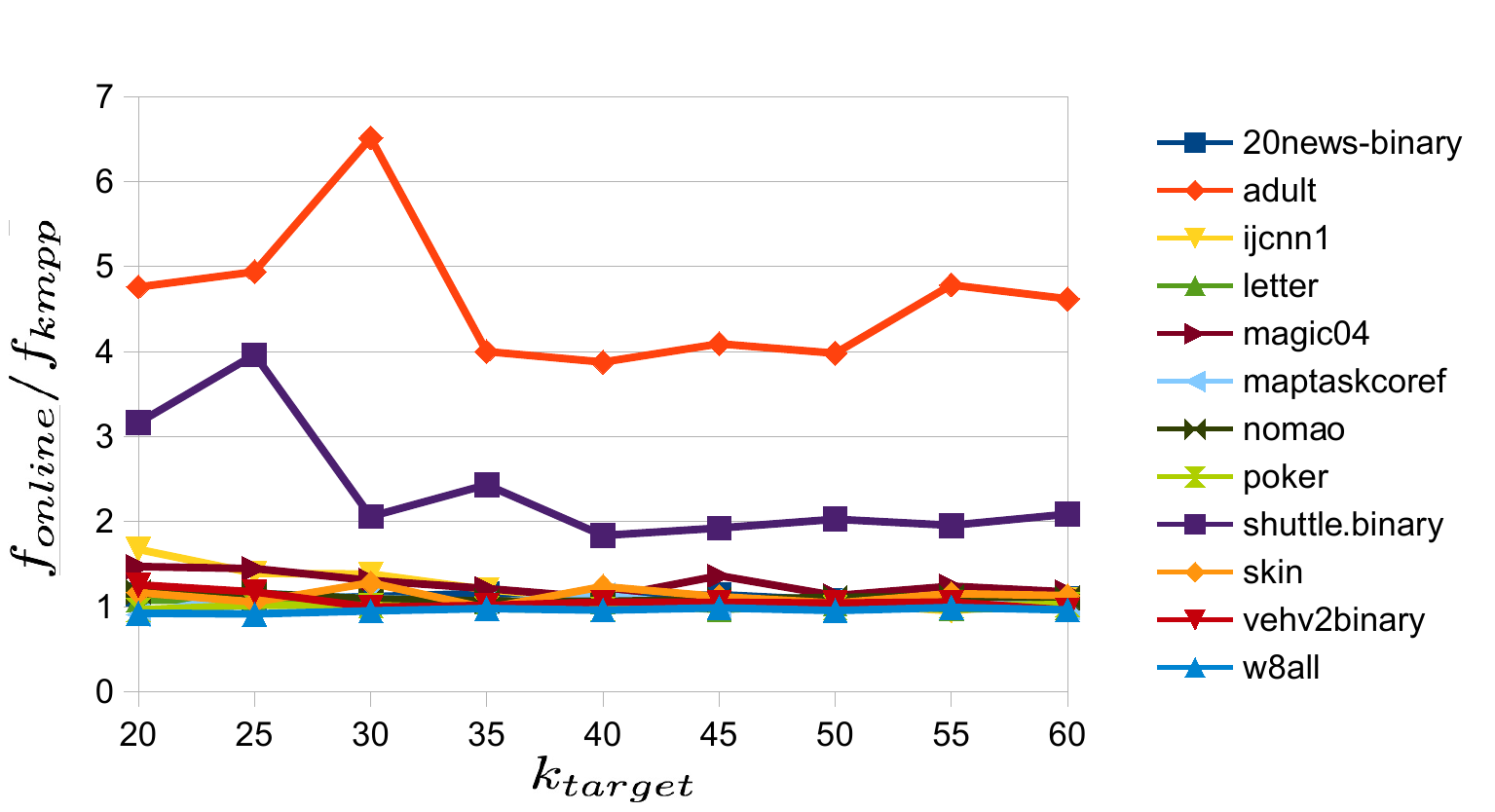}
\caption{On the $y$-axis we plot $f_{online}/f_{kmpp}$ as a function of $k_{target}$ on the $x$-axis. The values of $f_{online}$ is the cost of running Algorithm~\ref{algExperimental} with parameter $k_{target}$.
The value of $f_{kmpp}$ is the cost of running $k$-means++ with $k_{actual}$ clusters, $k_{actual}$ is the number of clusters online $k$-means actually used.
We see that, except for the datasets \emph{adult} and \emph{shuttle.binary}, $k$-means++ and online $k$-means are comparable. For \emph{adult} and \emph{shuttle.binary} online $k$-means is worse by a small constant factor. Note (Figure~\ref{fig3}) that both \emph{adult} and \emph{shuttle.binary} are datasets for which online $k$ means is dramatically better than random.}
\label{fig4}
\end{center}
\end{figure}

\begin{figure}[h!]
\begin{center}
\includegraphics[width=\myFigureWidth]{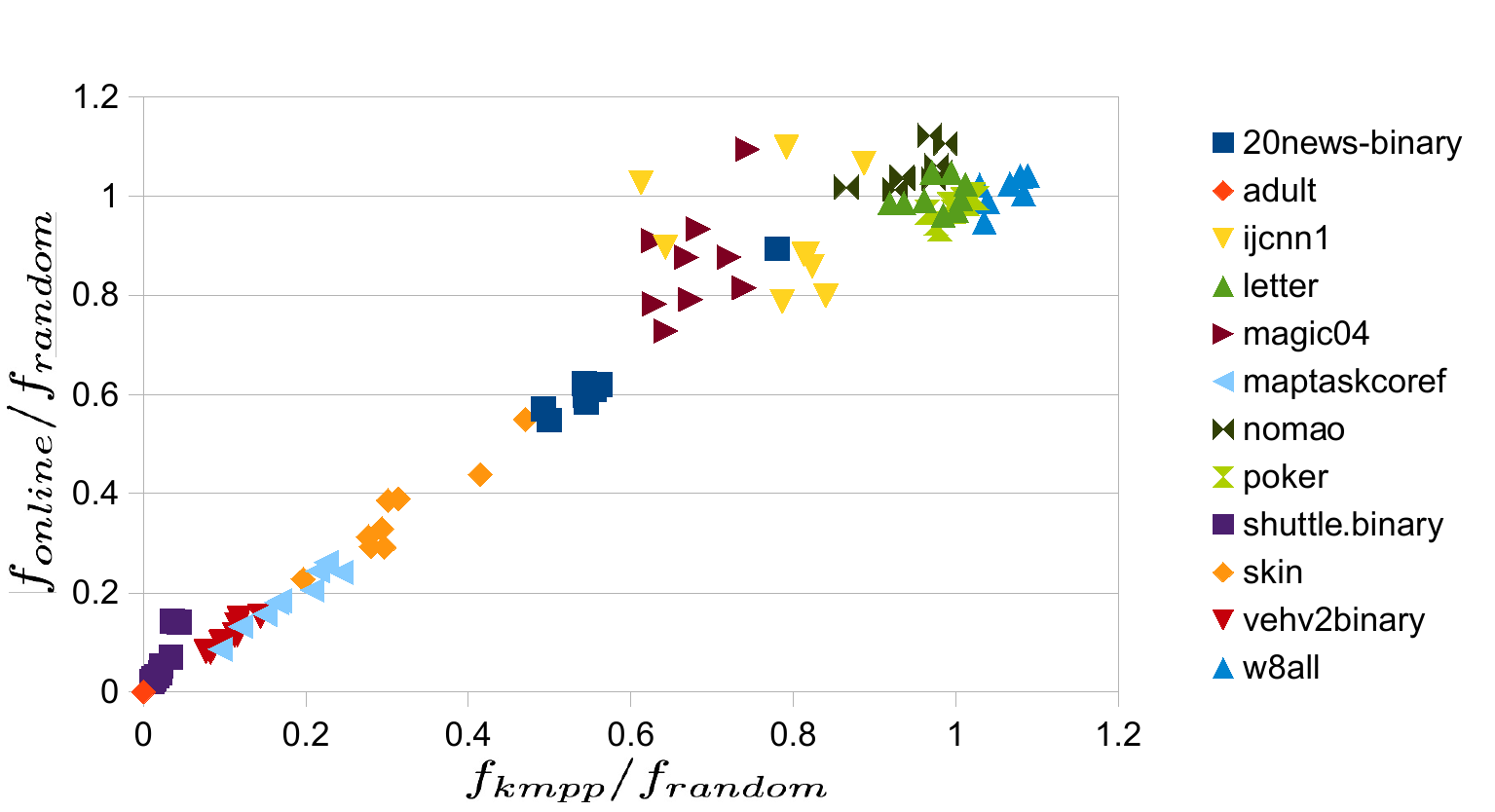}
\caption{On the $x$-axis $f_{kmpp}/f_{random}$ both using $k_{actual}$ clusters. 
The value of $k_{actual}$ is obtained by running online $k$-means with input $k_{target}$ on the $x$-axis.
The $y$-axis depicts $f_{online}/f_{random}$. 
Note that the performance of $k$-means++ and online $k$-means are very similar almost everywhere.
The advantage of $k$-means++ (see Figure~\ref{fig4}) occurs when the clustering cost is a minuscule fraction of random clustering.}
\label{fig5}
\end{center}
\end{figure}
  
\section{Aknowledgements}
We would like to thank Anna Choromanska and Sergei Vassilvitskii for very helpful suggestions and to Dean Foster for helping us with the proof of the Lemma \ref{random}.

\bibliographystyle{plain}

\begin{thebibliography}{10}

\bibitem{AilonJM09}
Nir Ailon, Ragesh Jaiswal, and Claire Monteleoni.
\newblock Streaming k-means approximation.
\newblock In Yoshua Bengio, Dale Schuurmans, John~D. Lafferty, Christopher
  K.~I. Williams, and Aron Culotta, editors, {\em NIPS}, pages 10--18. Curran
  Associates, Inc., 2009.

\bibitem{ArthurV07}
David Arthur and Sergei Vassilvitskii.
\newblock k-means++: the advantages of careful seeding.
\newblock In Nikhil Bansal, Kirk Pruhs, and Clifford Stein, editors, {\em
  SODA}, pages 1027--1035. SIAM, 2007.

\bibitem{AryaGKMMP04}
Vijay Arya, Naveen Garg, Rohit Khandekar, Adam Meyerson, Kamesh Munagala, and
  Vinayaka Pandit.
\newblock Local search heuristics for k-median and facility location problems.
\newblock {\em SIAM J. Comput.}, 33(3):544--562, 2004.

\bibitem{UCIdata}
K.~Bache and M.~Lichman.
\newblock {UCI} machine learning repository, 2013.

\bibitem{ChoromanskaM12}
Anna Choromanska and Claire Monteleoni.
\newblock Online clustering with experts.
\newblock In Neil~D. Lawrence and Mark Girolami, editors, {\em AISTATS},
  volume~22 of {\em JMLR Proceedings}, pages 227--235. JMLR.org, 2012.

\bibitem{CoatesNL11}
Adam Coates, Andrew~Y. Ng, and Honglak Lee.
\newblock An analysis of single-layer networks in unsupervised feature
  learning.
\newblock In Geoffrey~J. Gordon, David~B. Dunson, and Miroslav Dud\'{\i}k,
  editors, {\em AISTATS}, volume~15 of {\em JMLR Proceedings}, pages 215--223.
  JMLR.org, 2011.

\bibitem{DasguptaCSE291}
Sanjoy Dasgupta.
\newblock Topics in unsupervised learning.
\newblock Class Notes CSE 291.

\bibitem{drezner2004facility}
Z.~Drezner and H.W. Hamacher.
\newblock {\em Facility Location: Applications and Theory}.
\newblock Springer series in operations research. Springer, 2004.

\bibitem{libsvmData}
Rong-En Fan.
\newblock Libsvm data: Classification, regression, and multi-label.
  http://www.csie.ntu.edu.tw/~cjlin/libsvmtools/datasets/.

\bibitem{Fotakis08}
Dimitris Fotakis.
\newblock On the competitive ratio for online facility location.
\newblock {\em Algorithmica}, 50(1):1--57, 2008.

\bibitem{Fotakis11}
Dimitris Fotakis.
\newblock Online and incremental algorithms for facility location.
\newblock {\em SIGACT News}, 42(1):97--131, 2011.

\bibitem{KanungoMNPSW02}
Tapas Kanungo, David~M. Mount, Nathan~S. Netanyahu, Christine~D. Piatko, Ruth
  Silverman, and Angela~Y. Wu.
\newblock A local search approximation algorithm for k-means clustering.
\newblock In {\em Symposium on Computational Geometry}, pages 10--18, 2002.

\bibitem{Lloyd82}
S.~Lloyd.
\newblock Least squares quantization in pcm.
\newblock {\em IEEE Trans. Inf. Theor.}, 28(2):129--137, September 2006.

\bibitem{Meyerson01}
Adam Meyerson.
\newblock Online facility location.
\newblock In {\em FOCS}, pages 426--431. IEEE Computer Society, 2001.

\bibitem{SWM}
Adam Meyerson, Michael Shindler, and Alex Wong.
\newblock Fast and accurate k-means for large datasets.
\newblock {\em NIPS}, 2011.

\bibitem{OstrovskyRSS12}
Rafail Ostrovsky, Yuval Rabani, Leonard~J. Schulman, and Chaitanya Swamy.
\newblock The effectiveness of lloyd-type methods for the k-means problem.
\newblock {\em J. ACM}, 59(6):28, 2012.

\bibitem{Vygen05}
Jens Vygen.
\newblock Approximation algorithms for facility location problems.
\newblock Lecture Notes, Technical Report No. 05950, 2005.

\end{thebibliography}


\begin{thebibliography}{10}

\bibitem{AckermannMRSLS12}
Marcel~R. Ackermann, Marcus M{\"{a}}rtens, Christoph Raupach, Kamil Swierkot,
  Christiane Lammersen, and Christian Sohler.
\newblock Streamkm++: {A} clustering algorithm for data streams.
\newblock {\em {ACM} Journal of Experimental Algorithmics}, 17(1), 2012.

\bibitem{AggarwalDK09}
Ankit Aggarwal, Amit Deshpande, and Ravi Kannan.
\newblock Adaptive sampling for k-means clustering.
\newblock In {\em Approximation, Randomization, and Combinatorial Optimization.
  Algorithms and Techniques, 12th International Workshop, {APPROX} 2009, and
  13th International Workshop, {RANDOM} 2009, Berkeley, CA, USA, August 21-23,
  2009. Proceedings}, pages 15--28, 2009.

\bibitem{AilonJM09}
Nir Ailon, Ragesh Jaiswal, and Claire Monteleoni.
\newblock Streaming k-means approximation.
\newblock In Yoshua Bengio, Dale Schuurmans, John~D. Lafferty, Christopher
  K.~I. Williams, and Aron Culotta, editors, {\em NIPS}, pages 10--18. Curran
  Associates, Inc., 2009.

\bibitem{AnagnostopoulosBUH04}
Aris Anagnostopoulos, Russell Bent, Eli Upfal, and Pascal~Van Hentenryck.
\newblock A simple and deterministic competitive algorithm for online facility
  location.
\newblock {\em Inf. Comput.}, 194(2):175--202, 2004.

\bibitem{ArthurV07}
David Arthur and Sergei Vassilvitskii.
\newblock k-means++: the advantages of careful seeding.
\newblock In Nikhil Bansal, Kirk Pruhs, and Clifford Stein, editors, {\em
  SODA}, pages 1027--1035. SIAM, 2007.

\bibitem{AryaGKMMP04}
Vijay Arya, Naveen Garg, Rohit Khandekar, Adam Meyerson, Kamesh Munagala, and
  Vinayaka Pandit.
\newblock Local search heuristics for k-median and facility location problems.
\newblock {\em SIAM J. Comput.}, 33(3):544--562, 2004.

\bibitem{UCIdata}
Kevin Bache and Moshe Lichman.
\newblock {UCI} machine learning repository, 2013.

\bibitem{BahmaniMVKV12}
Bahman Bahmani, Benjamin Moseley, Andrea Vattani, Ravi Kumar, and Sergei
  Vassilvitskii.
\newblock Scalable k-means++.
\newblock {\em {PVLDB}}, 5(7):622--633, 2012.

\bibitem{CharikarCFM97}
Moses Charikar, Chandra Chekuri, Tom\'{a}s Feder, and Rajeev Motwani.
\newblock Incremental clustering and dynamic information retrieval.
\newblock In {\em Proceedings of the Twenty-ninth Annual ACM Symposium on
  Theory of Computing}, STOC '97, pages 626--635, New York, NY, USA, 1997. ACM.

\bibitem{ChoromanskaM12}
Anna Choromanska and Claire Monteleoni.
\newblock Online clustering with experts.
\newblock In {\em Proceedings of the Fifteenth International Conference on
  Artificial Intelligence and Statistics, {AISTATS} 2012, La Palma, Canary
  Islands, April 21-23, 2012}, pages 227--235, 2012.

\bibitem{CoatesNL11}
Adam Coates, Andrew~Y. Ng, and Honglak Lee.
\newblock An analysis of single-layer networks in unsupervised feature
  learning.
\newblock In Geoffrey~J. Gordon, David~B. Dunson, and Miroslav Dud\'{\i}k,
  editors, {\em AISTATS}, volume~15 of {\em JMLR Proceedings}, pages 215--223.
  JMLR.org, 2011.

\bibitem{DasguptaCSE291}
Sanjoy Dasgupta.
\newblock Topics in unsupervised learning.
\newblock Class Notes CSE 291, 2014.

\bibitem{DreznerH02}
Zvi Drezner and Horst~W. Hamacher.
\newblock {\em Facility location - applications and theory}.
\newblock Springer, 2002.

\bibitem{libsvmData}
Rong-En Fan.
\newblock Libsvm data: Classification, regression, and multi-label., 2014.

\bibitem{Fotakis08}
Dimitris Fotakis.
\newblock On the competitive ratio for online facility location.
\newblock {\em Algorithmica}, 50(1):1--57, 2008.

\bibitem{Fotakis11}
Dimitris Fotakis.
\newblock Online and incremental algorithms for facility location.
\newblock {\em SIGACT News}, 42(1):97--131, 2011.

\bibitem{GuhaMMMO03}
Sudipto Guha, Adam Meyerson, Nina Mishra, Rajeev Motwani, and Liadan
  O'Callaghan.
\newblock Clustering data streams: Theory and practice.
\newblock {\em {IEEE} Trans. Knowl. Data Eng.}, 15(3):515--528, 2003.

\bibitem{KanungoMNPSW02}
Tapas Kanungo, David~M. Mount, Nathan~S. Netanyahu, Christine~D. Piatko, Ruth
  Silverman, and Angela~Y. Wu.
\newblock A local search approximation algorithm for k-means clustering.
\newblock In {\em Symposium on Computational Geometry}, pages 10--18, 2002.

\bibitem{LiangK09}
Percy Liang and Dan Klein.
\newblock Online {EM} for unsupervised models.
\newblock In {\em Human Language Technologies: Conference of the North American
  Chapter of the Association of Computational Linguistics, Proceedings, May 31
  - June 5, 2009, Boulder, Colorado, {USA}}, pages 611--619, 2009.

\bibitem{Lloyd82}
Stuart~P. Lloyd.
\newblock Least squares quantization in pcm.
\newblock {\em IEEE Trans. Inf. Theor.}, 28(2):129--137, September 1982.

\bibitem{Meyerson01}
Adam Meyerson.
\newblock Online facility location.
\newblock In {\em FOCS}, pages 426--431. IEEE Computer Society, 2001.

\bibitem{SWM}
Adam Meyerson, Michael Shindler, and Alex Wong.
\newblock Fast and accurate k-means for large datasets.
\newblock {\em NIPS}, 2011.

\bibitem{OstrovskyRSS12}
Rafail Ostrovsky, Yuval Rabani, Leonard~J. Schulman, and Chaitanya Swamy.
\newblock The effectiveness of lloyd-type methods for the k-means problem.
\newblock {\em J. ACM}, 59(6):28, 2012.

\bibitem{Vygen05}
Jens Vygen.
\newblock Approximation algorithms for facility location problems.
\newblock Lecture Notes, Technical Report No. 05950, 2005.

\end{thebibliography}

\end{document}